\documentclass{llncs}

\usepackage{hyperref}
\usepackage{graphicx}
\usepackage{xcolor}
\usepackage{subcaption}
\usepackage{amsmath}
\usepackage{amsfonts}
\usepackage{cite}
\usepackage[title]{appendix} %CHECK
\usepackage[misc]{ifsym}

\usepackage{marginnote}

\begin{document}
\title{Reliable Broadcast in Dynamic Networks with Locally Bounded Byzantine Failures \thanks{This work was performed within Project ESTATE (Ref. ANR-16-CE25-0009-03), supported by French state funds managed by the ANR (Agence Nationale de la Recherche), and it has has been partially supported by the INOCS Sapienza Ateneo 2017 Project (protocol number RM11715C816CE4CB).	Giovanni Farina thanks the \emph{Universit\'e Franco-Italienne/Universit\'a Italo-Francese} (UFI/UIF) for supporting his mobility through the Vinci grant 2018.}}
\author{Silvia Bonomi\inst{1} \and
	Giovanni Farina\inst{2,1} \Letter \and
	S\'ebastien Tixeuil\inst{2}}
\institute{
	Dipartimento di Ingegneria Informatica Automatica e Gestionale Antonio Ruberti,\\
	Sapienza Universit\`a di Roma, Rome, Italy\\
	\email{bonomi@diag.uniroma1.it}
	\and
	Sorbonne Universit\'e,\\
	CNRS, Laboratoire d'Informatique de Paris 6, LIP6,
	F-75005 Paris, France
	\email{\{giovanni.farina,sebastien.tixeuil\}@lip6.fr}}

\maketitle              % typeset the header of the contribution
%
%\begin{center}
%	{\small\textbf{Regular paper}}
%\end{center}

\begin{abstract}
	Ensuring reliable communication despite possibly malicious participants is a primary objective in any distributed system or network. In this paper, we investigate the possibility of reliable broadcast in a dynamic network whose topology may evolve while the broadcast is in progress. 
	In particular, we adapt the Certified Propagation Algorithm (CPA) to make it work on dynamic networks and we present conditions (on the underlying dynamic graph) to enable safety and liveness properties of the reliable broadcast. We furthermore explore the complexity of assessing these conditions for various classes of dynamic networks. 	
	\keywords{Byzantine Reliable Broadcast \and Locally bounded failures \and Dynamic Networks.}
\end{abstract}

\section{Introduction}
% !TEX root = main.tex
\noindent
Designing dependable and secure systems and networks that are able to cope with various types of adversaries, ranging from simple errors to internal or external attackers, requires to integrate those risks from the very early design stages. The most general attack model in a distributed setting is the Byzantine model, where a subset of nodes participating in the system may behave arbitrarily (including  in a malicious manner), while the rest of processes remain correct. Also, reliable communication primitives are a core building block of any distributed software. Finally, as current applications are run for extended periods of time with expected high availability, it becomes mandatory to integrate dynamic changes in the underlying network while the application is running.  
In this paper, we address the reliable \emph{broadcast} problem (where a \emph{source} node must send data to every other node) in the context of dynamic networks (whose topology may change while the broadcast is in progress) that are subject to Byzantine failures (a subset of the nodes may act arbitrarily).
The reliable broadcast primitive is expected to provide two guarantees: \emph{(i)} \emph{safety}, namely if a message $m$ is delivered by a correct process, then $m$ was sent by the source and \emph{(ii)} \emph{liveness}, namely if a message $m$ is sent by the source, it is eventually delivered by every correct process.

\noindent\textbf{Related Works.}
%\GFN{EDITED, removed the reference to the complete networks, that results is related to the byzantine agreement. This paragraph has been modified. I made slight modifications over all the subsection}
%
%In static networks (that is, in networks whose topology remains fixed during the entire execution of the protocol), a solution for the reliable broadcast problem has been initially provided for complete networks \cite{lamport1982byzantine}, requiring that no more than one third of the nodes are Byzantine (that is, $n>3f$, where $n$ denotes the size of the network ad $f$ the maximum number of Byzantine nodes). 
In static multi-hop networks (in which the topology remains fixed during the entire execution of the protocol) the necessary and sufficient condition enabling reliable broadcast while the maximum number of Byzantine failure is bounded by $f$ has been identified by Dolev~\cite{dolev1981unanimity}, stating that this problem can be solved if and only if the network is $2f+1$-connected. 
Subsequently, the reliable broadcast problem has been analyzed assuming a local condition on the number of Byzantine neighbors a node may have~\cite{koo2004broadcast,pelc2005broadcasting}.
%This global condition on the number of failures was replaced by a local condition on the number of Byzantine neighbors a node may have~\cite{koo2004broadcast,pelc2005broadcasting}.
All aforementioned works require high network connectivity.
Indeed, extending a reliable broadcast service to sparse networks required to weaken the achieved guarantees~\cite{DBLP:journals/ppl/MaurerT16, DBLP:journals/tpds/MaurerT15, DBLP:journals/jpdc/MaurerT14}: \emph{(i)} accepting that a small minority of correct nodes may accept invalid messages (thus compromising safety), or accepting that a small minority of correct nodes may not deliver genuine messages (thus compromising liveness).

Adapting to dynamic networks proved difficult, as the topology assumptions made by the mentioned proposals may no longer hold: the network changes during the execution. 
Some core problems of distributed computing have been considered in the context of dynamic networks subject to Byzantine failures~\cite{DBLP:conf/podc/GuerraouiHK13, DBLP:conf/podc/AugustinePR13} but, to the best of our knowledge, there exists a single contribution for the reliable communication problem, due to Maurer \emph{et al.}~\cite{maurer2015communicating}. Their work can be seen as the dynamic network extension of the Dolev~\cite{dolev1981unanimity} solution for static networks, and assumes that no more than $f$ Byzantine processes {are present} in the network. Also, the protocol to be executed {spreads an exponential number of messages with respect to the size of the network} and requires each node to compute the minimal cut over the set of paths traversed by each received message, 
making the protocol unpractical 
for real applications.

The Byzantine tolerant reliable broadcast can also be solved by employing cryptography (e.g., digital signatures) \cite{castro1999practical, drabkin2005efficient} that enable all nodes to exchange messages guaranteeing authentication and integrity. 
The main advantage of cryptographic protocols is that they allow solving the problem with simpler solutions and weaker conditions (in terms of connectivity requirements). However, on the negative side, the safety of the protocols is bounded to the crypto-system.

\noindent\textbf{Contributions.}
In this paper, we investigate the possibility of reliable broadcast in a dynamic network that is subject to Byzantine faults. 
More precisely, we address the possibility of a local criterion on the number of Byzantine (as opposed to a global criterion as in Maurer \emph{et al.}~\cite{maurer2015communicating}) in the hope
that a practically efficient protocol can be derived in case the criterion is satisfied. Our starting point is the CPA protocol~\cite{koo2004broadcast,pelc2005broadcasting,DBLP:journals/ipl/TsengVB15,DBLP:journals/tpds/BhandariV10}, that was originally designed for static networks. In particular, our contributions can be summarized as follows: 
	 \emph{(i)} we extend the CPA algorithm to make it work in dynamic networks;
	 \emph{(ii)} we prove that the original safety property of CPA naturally extends to 
	 {dynamic networks}
	 and we define new liveness conditions 
	 specifically suited for the dynamic networks and
	 \emph{(iii)} we investigate the impact of nodes awareness about the dynamic network on reliable broadcast possibility and efficiency.
%
%Due to lack of space, part of the proofs of lemmas and theorems are omitted. They can be found inside the full version paper \\ \url{https://hal.archives-ouvertes.fr/hal-01712277}
%\ST{to be extended if space allows}

\section{System Model \& Problem Statement}\label{sec:sys_mod}
% !TEX root = main.tex
\setlength{\parskip}{0em}
We consider a distributed system composed by a set of $n$ processes $\Pi=\{p_1, p_2,\\ \dots p_n\}$, each one having a unique integer identifier. 
The passage of time is measured according to a fictional global clock spanning over natural numbers $\mathbb{N}$. 
The processes are arranged in a multi-hop communication network. The network can be seen as an undirected graph where each node represents a process $p_i \in \Pi$ and each edge represents a communication channel between two elements $p_i, p_j \in \Pi$ such that $p_i$ and $p_j$ can communicate. \\
\noindent\textbf{Dynamic Network Model.} The communication network is \emph{dynamic} i.e., the set of edges (or available communication channels) changes over time.
More formally, we model the network as a \emph{Time Varying Graph} (TVG) \cite{casteigts2012time} i.e., a graph $\mathcal{G}=( V,E,\rho,\zeta)$ where:
\begin{itemize}
	\item $V$ is the set of processes (in our case $V=\Pi$);
	\item $E \subseteq V \times V$ is the set of edges (i.e., communication channels).
	\item $\rho : E \times \mathbb{N} \rightarrow \{0,1\}$ is the \emph{presence} function. Given an edge $e_{i, j}$ between two nodes $p_i$ and $p_j$, $\rho(e_{i, j},t) = 1$ indicates that edge $e_{i, j}$ is present at time $t$;
	\item $\zeta : E \times \mathbb{N} \rightarrow \mathbb{N}$ is the \emph{latency} function that indicates how much time is needed to cross an edge starting from a given time $t$. In particular, $\zeta(e_{i, j},t) = \delta_{i, j,t}$ indicates that a message $m$ sent at time $t$ from $p_i$ to $p_j$ takes $\delta_{i, j,t}$ time units to cross edge $e_{i, j}$.
\end{itemize}
The evolution of $\mathcal{G}$ can also be described as a sequence of static graphs $\mathcal{S_{\mathcal{G}}} = G_0, G_1, \dots G_T$ where $G_i$ corresponds to the \emph{snapshot} of $\mathcal{G}$ at time $t_i$
(i.e. $G_i=(V, E_i)$ where $E_{i}=\{e \in E~|~ \rho(e,t_i) = 1\}$).	
No further assumptions on the evolution of the dynamic network are made.
The static graph $G = (V, E)$ that considers all the processes and all the possible existing edges is called \emph{underlying graph} of $\mathcal{G}$ and it flattens the time dimension indicating only the pairs of nodes that have been connected at some time $t'$.
In the following, we interchangeably use terms \emph{process} and \emph{node} and we will refer to \emph{edges} and \emph{communication channels} interchangeably.
Let us note that the TVG model is one among the most general available and it is able to abstract and characterize several real dynamic networks \cite{casteigts2012time}.

\noindent\textbf{Communication model and Timing assumption.} 
%\GFN{slight modifications}
%
%Every process is able to communicate only with its direct neighbors in the graph through message passing. We call \emph{sender of a message $m$} the process $p$ that transmits $m$.
%{
Processes communicate through message exchanges. Every message has (i) a \emph{source}, which is the id of the process that has created the message %and it is encapsulated in the message content,
and (ii) a \emph{sender}, that is the id of the process that is relaying the message. 
The source and the sender may coincide.  
The sender is always a neighbor in the communication network. 
%\GFN{I added this explicit reference to the source ID}
The ID of the source is included inside the message, i.e. any message is composed by its content and the source ID.
We refer with $m_s$ to a message $m$ with $p_s$ as source.

We assume \emph{authenticated} and \emph{reliable} point-to-point channels
where (a) \emph{authenticated} ensures that the identity of the sender cannot be forged; (b) \emph{reliable} guarantees that the channel delivers a message $m$ if and only if (i) $m$ was previously sent by %$p$
%\GFN{I changed the reliable definition, (b)-(i)}
its sender and (ii) the channel has been up long enough to allow the reception (i.e. given a message $m$ sent at time $t$ from $p_i$ to $p_j$ and having latency $\delta_{i, j, t}$, we will have reliable delivery if $\rho(e_{i, j},\tau) = 1$ for each $\tau \in [t, t+\delta_{i, j, t}]$).
Notice that these channel assumptions are implicitly made also on analysis of CPA on static networks
and that they are both essential to guarantees the reliable broadcast properties. 

At every time unit $t$ each process takes the following actions: (i) \emph{send} where processes send all the messages for the current time unit (potentially none), (ii) \emph{receive} where processes receive and store all the messages for the current time unit (potentially none) and (iii) \emph{computation} where processes process the buffer of received messages and compute the messages to be sent during the next time unit according to the deterministic distributed protocol $\mathcal{P}$ that they are executing. 
%\NOTE{"It is not clear when a computing step occurs"}[-2cm]
%$\mathcal{P}$ is represented by a finite state automata and it is composed of a sequence of computation and communication steps. A computation step is represented by the computation executed locally to each process while a communication step is represented by the sending and the delivering events of a message. Computation steps and communication steps are generally called \emph{events}.
Thus, the system is assumed to be synchronous in the sense that (i) every channel has a latency function that is bounded and the overall message delivery time is bounded by the maximum channel latency and (ii) computation steps are bounded by a constant that is negligible with respect to the overall message delivery time and we consider it equal to $0$.
%and evolves in sequential synchronous rounds $r_0, r_1,$ $\dots r_i \dots$, $r_i \in \mathbb{N}$.  
%Without loss of generality, we assume that the granularity of the clock is such to associate a round to each time unit.
%}
%\NOTE{"You say that your system is synchronous, but you do not tell what this really means."}[-2cm]
%{\color{blue}
%Every process iteratively execute the following steps on every time instant: (1) it receives messages, (2) it elaborates them and (3) it sends messages, according to the broadcast protocol.
%}
%We assume that the computation time is negligible with respect to communication and we consider it equal to $0$.
%Let us note 
%{\color{blue}
%the time bounds characterized by latency function $\zeta(e,t)$ of $TVG$ and our assumption on the computation time 
%model a dynamic \emph{synchronous} system. 
We discuss the implications and consequences of lack of synchrony inside the full version paper.
% in the Appendix \ref{app:asy}.
%

\noindent\textbf{Failure model.} 
%\NOTE{"f-locally \dots this assumption is realistic?"}[.5cm]
%\GFN{this scenario applies to all kinds
%	of networks, where each node is assumed to be able to estimate the number of traitors in its close neighborhood (social networks)}[1.5cm]
We assume an omniscient adversary able to control several processes of the network allowing them to behave arbitrarily (including corrupting/dropping messages or simply crashing). We call them \emph{Byzantine} processes.
Processes that are not Byzantine faulty are said to be \emph{correct}. Correct processes do not a priori know which processes are Byzantine.
%\GFN{rephrased}
Specifically to reliable broadcast protocols, a Byzantine process can spread messages carrying a fake source ID and/or content or it can drop any received message preventing its propagation.

We considered the \emph{f-locally bounded} failure model \cite{koo2004broadcast} as all CPA related works, i.e., along time every process $p_i$ can be connected with at most $f$ Byzantine processes.
In other words, given the underlying static graph $G=(V, E)$, every process $p_i \in V$ has at most $f$ Byzantine neighbors in $G$.\\
\noindent{\bf Problem Statement.}
%\NOTE{"Why can you restrict your attention to correct sources, given the fact that no receiver can know whether the source is correct or not? (On the other hand: is it really the case that the correct source assumption is needed for safety? Given (i), it should be possible to guarantee safety also for Byzantine sources, no?)"}[1cm]
%\GF{On a review it is said that words "broadcast, send, forward and deliver" are too technical, should we place further explanations? the meaning of broadcast is stated in the introduction}
%\ST{No, this is for PODC, they know.}
%
In this paper, we consider the problem of \emph{Reliable Broadcast over dynamic networks} assuming a $f$-locally bounded Byzantine failure model from a given \emph{correct} \footnote{note the assumption of a possibly faulty source leads to a more general problem, the \emph{Byzantine Agreement} \cite{dolev1981unanimity}} source $p_s$.
We say that a protocol $\mathcal{P}$ satisfies \emph{reliable broadcast}, if a message $m$ broadcast by a correct process $p_s \in \Pi$ (also called \emph{source} or \emph{author}) is eventually delivered
(i.e., accepted as a valid message)
by every correct process $p_j \in \Pi$. 
Said differently, a protocol $\mathcal{P}$ satisfies \emph{reliable broadcast}, if the following conditions are met:
\begin{itemize}
	\item {\bf Safety} if a message $m$ is delivered by a correct process, then such message has been sent by {the source} $p_s$;
	\item {\bf Liveness}: if a message $m$ is broadcast by {the source} $p_s$, it is eventually delivered by every correct process.
\end{itemize}
In other words, a reliable broadcast protocol extends the guarantees provided by the communication channels to the message exchanges between a node and any correct process not directly connected to it.

\section{The Certified Propagation Algorithm (CPA)}
\label{sec:CPA}
%\NOTE{"Is the source $p_s$ known to all receivers a priori or is it flagged as such in the messages?"}
%\GFN{I'm changed the message notation from $m$ to $m_s$}[3cm]
\noindent
The Certified Propagation Algorithm (CPA) \cite{koo2004broadcast,pelc2005broadcasting} is a protocol enforcing reliable broadcast, from a correct source $p_s$, in \emph{static} multi-hop networks with a $f$-locally bounded Byzantine adversary model, where nodes have no knowledge on the global network topology. 
Given a message $m$ to be broadcast, CPA starts the propagation of $m_s$ from $p_s$ and applies three acceptance policies (denoted by \emph{AC}) to decide if $m_s$ should be accepted and forwarded (\emph{i.e.}, transmitted also by nodes different from the source) by a process $p_j$.
Specifically:
\begin{itemize}
	\item [-] $p_s$ delivers $m_s$ {\bf (AC1)}, forwards it to all of its neighbors, and stops; 
	\item [-] when receiving $m_s$ from $p_i$, if $p_i$ is the source then $p_j$ delivers $m_s$ {\bf (AC2)}, forwards $m_s$ to all of its neighbors and stops; otherwise 
	%($p_j$ is not a neighbor of $p_s$), 
	the message is buffered. 
	\item [-]  upon receiving $f+1$ copies of $m_s$ from distinct neighbors, $p_j$ delivers $m_s$ {\bf (AC3)}, then forwards it to all its neighbors and stops. 
\end{itemize}

The correctness of CPA on static networks has been proved to be dependent on the network topology. In particular, Litsas \emph{et al.}~\cite{litsas2013graph} provided topological conditions based on the concept of \emph{$k$-level ordering}. Informally, given a graph $G=(V, E)$ and considering a node $p_s$ as the source, we can define a $k$-level ordering as a partition of nodes into \emph{ordered levels} such that: \emph{(i)} $p_s$ belongs to level $L_0$, \emph{(ii)} all the neighbors of $p_s$ belong to level $L_1$, and \emph{iii)} each node in a level
%\NOTE{"you should say that the $k$ neighbors may be from different levels"}
$L_i$ has at least $k$ neighbors {over} levels $L_j$, with $j<i$.
%\GFN{I removed the MKLO example on static networks}
A $k$-level ordering is \emph{minimum} if every node appears in the minimum level possible.
\begin{definition}[MKLO]
Let $G=(V, E)$ be a graph and let $p_s$ be a node of $G$ called \emph{source}. 
%\GFN{The MKLO is now explicitly bounded to $p_s$.}
The \emph{minimum $k$-level ordering} \emph{(MKLO)} of $G$ from $p_s$ is the partition $P_k$ of nodes into disjoint subsets called levels $L_i$ defined as follows: 
$$
\begin{cases}
		p \in L_0 & $if $ p=p_s \\
		p \in L_1 & $if $ p \in N_s\\
		p \in L_{i>1} & $if $p \in V \setminus (\bigcup\limits_{j=0}^{i-1} L_j) $ and $ |N_p \cap (\bigcup\limits_{j=0}^{i-1} L_j )| \geq k\\
\end{cases}
$$
\end{definition}

%\begin{figure}[t]
%\centering
%\includegraphics[scale=0.35]{figures/K2LO}
%\caption{A graph $G=(V, E)$ allowing a $k$-level ordering with $k=3$.}\label{fig:suff}
%\end{figure}
%
%An example minimum $k$-level ordering (with $k=3$) is shown in Figure \ref{fig:suff}. A relaxed $k$-level ordering for the same graph can be obtained by placing $p_5$ in a level and $p_7$ in another. Let us note that multiple relaxed $k$-level ordering may exist while the minimum $k$-level ordering is unique. Also, once a relaxed $k$-level ordering is defined, it can be reduced to a minimum $k$-level ordering.

For CPA to ensure reliable broadcast from $p_s$, a sufficient condition is that a $k$-level ordering from $p_s$ exists, with $k \ge 2f+1$. Conversely, the necessary condition demands a $k$-level ordering from $p_s$ with $k \ge f+1$ (see \cite{litsas2013graph}).
%Considering Figure~\ref{fig:suff} and assuming $f=1$, CPA ensures reliable broadcast from $p_s$ as the sufficient condition is satisfied.
Those conditions can be verified with an algorithm whose time complexity is polynomial in the size of the network, specifically with a modified Breadth-First-Search.
In the case that a graph $G=(V, E)$ satisfies the necessary condition from $p_s$ but not the sufficient one, then further analysis must be carried out. In particular, in order to verify whether $G$ enables reliable broadcast from $p_s$, one should check whether a $k$-level ordering from $p_s$ exists (with $k=f+1$) in every sub-graph $G'$ obtained from $G$ by removing all nodes corresponding to possible Byzantine placement in the $f$-locally bounded assumption. 
The verification of the strict condition has been proven to be NP-Hard~\cite{ichimura2010new}.
\section{The Certified Propagation Algorithm on Dynamic Networks}
\label{sec:dcpa}
%In this section, we consider how CPA behaves on dynamic networks, \emph{i.e.} networks whose topology may evolve in time, 
In this section, we analyze how CPA behaves on dynamic networks, \emph{i.e.} networks whose topology may evolve over time, 
and how it needs to be extended to work in such settings.

\begin{figure}
	\centering
	\begin{subfigure}{.6\textwidth}
		\centering
		\includegraphics[width=.95\linewidth]{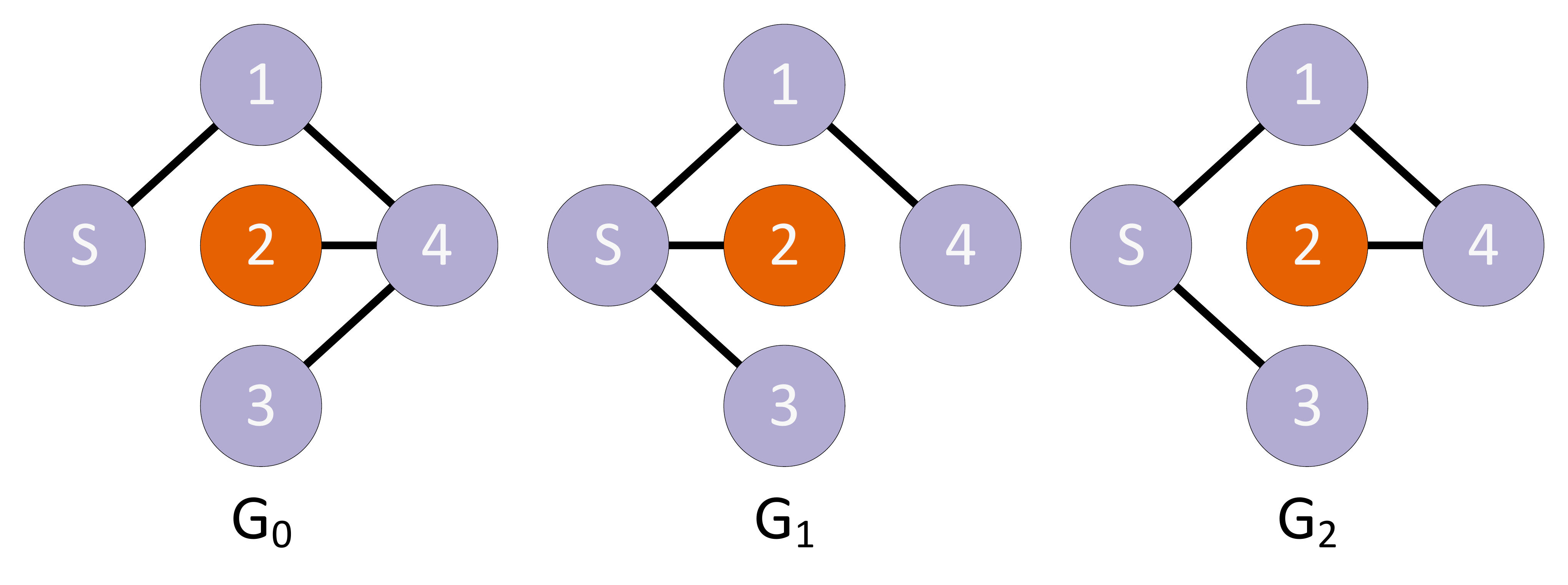}
		\caption{A Time Varying Graph $\mathcal{G}=( V,E,\rho,\zeta)$.}\label{sfig:tvg}
	\end{subfigure}%
	\begin{subfigure}{.4\textwidth}
		\centering
		\includegraphics[width=.47\linewidth]{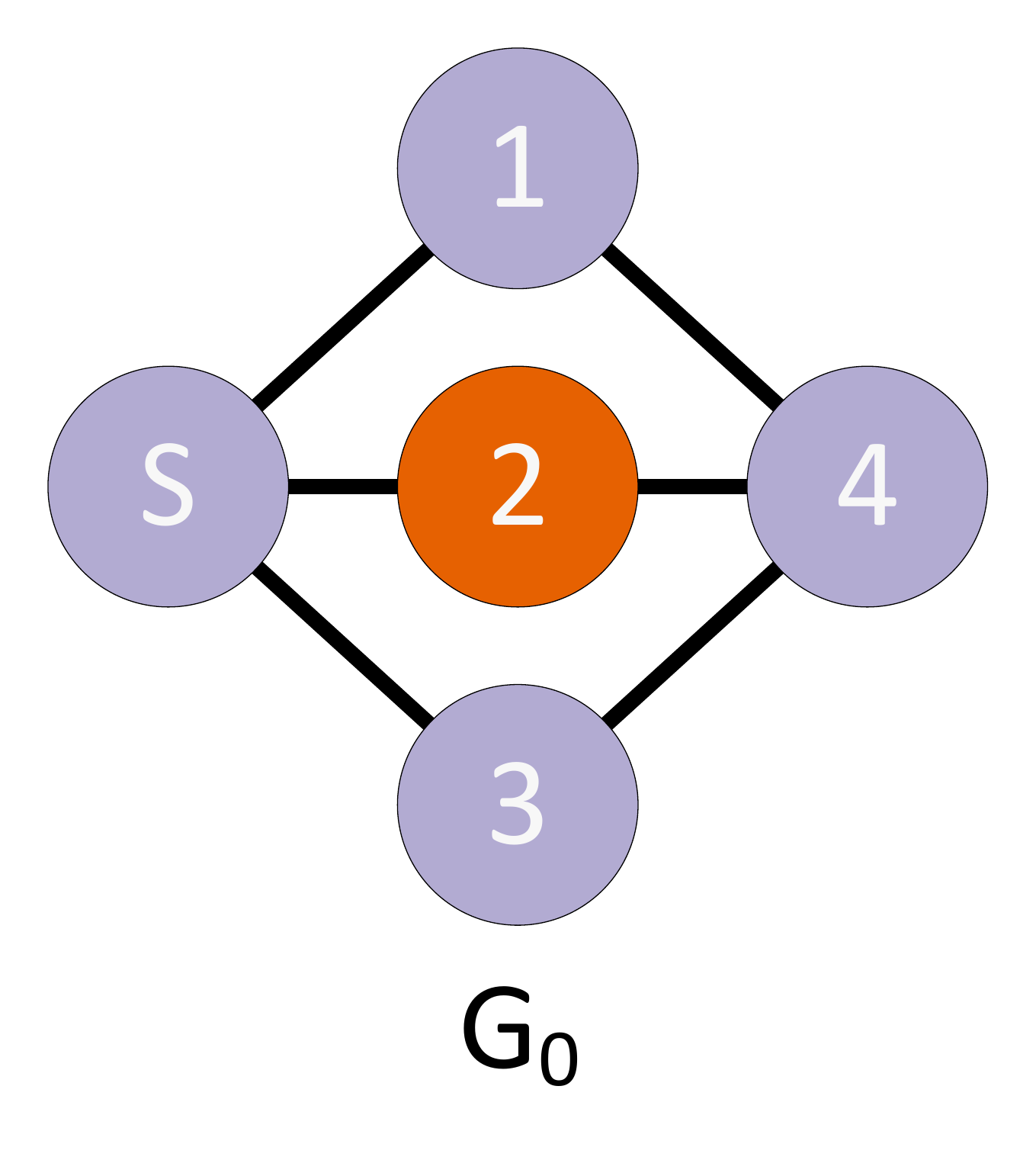}
		\caption{Underlying graph $G=(V, E)$.}\label{sfig:undGraph}
	\end{subfigure}
\caption{Example of a simple TVG and its underlying static graph.}\label{fig:TVG}
\end{figure}

Let us consider the TVG shown in Figure~\ref{fig:TVG} and suppose process $p_2$ is Byzantine. If we consider the static underlying graph $G=(V, E)$ shown in Figure~\ref{sfig:undGraph}, it is easy to verify that running CPA from the source node $p_s$ is possible to achieve reliable broadcast in a $1$-locally bounded adversary. However, if we consider snapshots of the TVG at different times\footnote{For the sake of simplicity, we consider the channel delay always equal to 1 in the example.} as shown in Figure~\ref{sfig:tvg}, one can verify that nodes $p_3$ and $p_4$ remain unable to deliver the message forever.  
In fact, $p_3$ is not a neighbor of the source $p_s$ when the message is broadcast by $p_s$ (\emph{i.e.}, at time $t_0$), and even if it had happened ($e_{s,3}$ at time $t_0$) the edge connecting $p_4$ with its correct neighbor $p_3$ appears only before the message would have been delivered and accepted by $p_3$, and thus it is not available for the retransmission. From this simple example its easy to see that the temporal dimension plays a fundamental role in the definition of topological constraints that a TVG must satisfy to enable reliable broadcast.

% CPA PSEUDOCODE
%\begin{figure}[t]
%	\centering	
%	\includegraphics[scale=0.8]{figures/DCPA}
%	\caption{DCPA (code for process $p_i$)}\label{fig:DCPA}
%\end{figure}
%REMOVED
%We propose modifications to CPA to extend its correctness to dynamic networks, naming the new algorithm Dynamic-CPA (DCPA).
%The pseudo code od DCPA is presented in Figure~\ref{fig:DCPA}. The main difference between DCPA and the original CPA lies in the introduction and management of a new set variable, $delivered$, where messages are stored together with their source in order to be retransmitted when edges change (line 4, line 9 and lines 27-31). This mechanism permits to processes to gain knowledge about transmitted messages.
%The acceptance policies (AC1,AC2,AC3) of CPA remain unchanged.
{\color{blue}
%We show that authenticated and reliable channels are necessary to ensure the reliable broadcast through CPA.
% thus the validity of this spreading algorithm demands the extension of such assumptions on dynamic networks.	
}
%PREVIOUS VERSION
%We first show that DCPA always ensures safety of reliable broadcast in dynamic networks. 
%That is, if a message $m$ is delivered by some correct process, then $m$ was sent by its source (see Theorem~\ref{th:correctness}). Indeed, this property is not affected by the dynamic change of edges but it only depends on the
%{\color{blue} reliable property}
%of channels.

%{\color{blue}
%Then, we consider a retransmission mechanics to let CPA works also on dynamic networks and analyze the conditions enabling liveness of reliable broadcast.
%} 

\subsection{CPA Safety in Dynamic Networks}
\label{subsec:cpasafe}
\noindent
In the following, we show that the authenticated and reliable channels are necessary to ensure the reliable broadcast through CPA.
%We first focus on the safety property of reliable broadcast.

%Recall that we are assuming, as made in the static case, reliable and authenticated channels. 
	
%A channel that (i) is authenticated and (ii) ensures no creation is necessary to achieve safety employing CPA.

\begin{lemma}\label{lem:safety}
	The CPA algorithm does not ensure safety of reliable broadcast when channels are not both authenticated and reliable (even on static graphs).
\end{lemma}
\begin{proof}
	An authenticated channel guarantees that the identity of the sender of a message cannot be forged. Without this assumption a Byzantine process can impersonate an arbitrary number of processes and invalidate the \emph{f-locally bounded} assumption. 

	A reliable channel guarantees that a message is received as it was sent by its sender. Without this assumption, an unreliable channel can potentially simulate a Byzantine process (namely the channel can deliver a message different from the one that was sent by the sender).
	\hfill
	$\Box$
\end{proof}

The same channel assumptions are sufficient for ensuring safety also on dynamic networks.
\begin{theorem}\label{th:correctness}
	Let $\mathcal{G}=( V,E,\rho,\zeta)$ be the TVG of a network with $f$-locally bounded Byzantine adversary. If every correct process $p_i$ runs CPA on top of reliable authenticated channels, then if a message $m_s$ is delivered by $p_i$, $m_s$ was previously sent by the correct source $p_s$.
\end{theorem}

\begin{proof}
The proof trivially follows from CPA correctness in static networks with $f$-locally bounded adversary, considering that in the underlying graph $G=(V, E)$, we still have a $f$-locally bounded adversary. \hfill$\Box$
\end{proof}

\subsection{CPA Liveness in Dynamic Networks}
\label{subsec:cpalive}
The CPA liveness in static networks is based on the availability of a certain topology that supports the message propagation. Indeed every edge is always up so, once the communication network satisfies the topological constraints imposed by the protocol, the assumption that channels do not lose messages is sufficient to guarantee their propagation.
%PREVIOUS VERSION
%CPA liveness in static networks is based on the availability of a certain topology that supports the message propagation assuming that messages are not lost.
%In static networks, in fact, every edge is always up so, once the communication network satisfies the topological constraints imposed by the protocol, assuming that channels do not loose messages is sufficient to guarantee their propagation.
In dynamic networks, this is no longer true.
Let us recall that each edge $e$ in a TVG is up according to its presence function $\rho(e, t)$. At the same time, the message delivery latency are determined by the edge latency function $\zeta(e, t)$. As a consequence, in order to ensure that a message $m$ sent at time $t$ from $p_i$ to $p_j$ is delivered, we need that $(p_i,p_j)$ remains up until time $t+\zeta(e, t)$.
%REMOVED
%(for ease of explanation and simplicity, remember that we also assumed timely channels, namely $ \forall e_{i, j} \in E, ~ \forall t \in \mathbb{N}, ~ \zeta(e_{i, j},t) \le \delta$). 
Contrarily, there could exist a communication channel where every message sent has no guarantee to be delivered as the edge disappears while the message is still traveling.
Thus, in addition to topological constraints, moving to dynamic networks we need to set up other constraints on when edges appear and for how long they remain up.
Considering that processes have no information about the network evolution, they do not know if and when a given transmitted message will reach its receiver. 
Hence, without assuming extra knowledge, a correct process must re-send messages infinitely often. 
% thus a retransmission mechanism has to be integrated inside the algorithm. 
%In particular the retransmission has to iterate infinitively often unless an additional mechanics guaranteeing a no dropping channel is employed or assumed.

As a consequence, CPA must be extended to the dynamic context incorporating the following additional steps:
\begin{itemize}
	\item [-] 
%	\GFN{rephrased} 
if process $p_i$ delivers a message $m$, it forwards $m$ to all of its neighbors infinitely often, at every time unit.
%	 In particular, $p_i$ will retransmit at every time unit all the messages $m$ received so far.
\end{itemize}
%\GFN{Here the part about infinite retransmission}
As a consequence, each time that the neighbors of $p_i$ changes, $p_i$ attempts to propagate the message.
Let us notice that such an infinite retransmission can be avoided/stopped only if a process get the acknowledgments about the delivery of the communication channels. This issue has been analyzed by considering further assumptions on the dynamic network
% guaranteeing the message deliveries over communication channels 
\cite{DBLP:conf/europar/Gomez-CalzadoCL15,DBLP:conf/aina/RaynalSCW14}.
To ease of explanation, we will refer to this extended version of CPA as Dynamic CPA (DCPA).

We now characterize the conditions enabling a communication channel to deliver messages in order to argue about liveness.
For this purpose, we define a boolean predicate whose value is ${\sf true}$ if and only if the TVG allows the reliable delivery of a message $m$ sent from $p_i$ to $p_j$ at time $t$.

\begin{definition}
Let $\mathcal{G}=( V,E,\rho,\zeta)$ be a TVG. We define the predicate \emph{Reliable Channel Delivery at time $t'$}, ${\sf RCD}(p_i, p_j, t')$ as follows:\\
$${\sf RCD}(p_i, p_j, t') =
\begin{cases}
 {\sf true} & if ~ \rho(<p_i, p_j>, \tau)=1, ~ \forall \tau \in [t', t' + \zeta(e_{i, j},t')].\\
 {\sf false} & otherwise.
\end{cases}
$$
\end{definition}
%\NOTE{"It is not clear what happens to a message where, during its latency, the presence function is 0 just once. Is it dropped altogether"}
{
The communication channels do not usually have memory, thus we consider any message sent while the ${\sf RCD}()$ predicate is {\sf false} as dropped. 
}

Now that we are able to express constraints on each edge through the ${\sf RCD}()$ predicate, we need to define those ${\sf RCD}()$ that enable liveness of reliable broadcast. Let us define the \emph{$k$-acceptance function}, that encapsulates temporal aspects for the three acceptance conditions of CPA.

%\NOTE{"The central Definition 3 is incomprehensible, and it it even appears to be cyclic"}[1cm]
\begin{definition}
	Let $p_s \in \Pi$ be a process that starts a reliable broadcast at time $t_{br}$. The \emph{$k$-acceptance function}
	$\mathcal{A}_k(p,t)$ over the time $t \in \mathbb{N}$ is defined as follows:
	%\begin{figure}[H]
	%	\centering
	%	\includegraphics[scale=0.90]{figures/acceptance_fun}
	%\end{figure}
	\vspace{0.3cm}
	
	\scalebox{0.85}{
		$\mathcal{A}_k(p_j,t)$ = 
		$\begin{cases} 
		1 & $ if $ p_j=p_s $ with $ t \ge t_{br} 
		~~~~~~~~~~~~~~~~~~~~~~~~~~~~~~~~~~~~~~~~~~
		~~~~~~
		 \text{(AK1)} \\
		%1 & $ if $ \exists~ t' \in [t_{br}, t] ~:~ {\sf RCD}(p_s, p, t')= {\sf true}\\ 
		1 & $ if $\exists~ t' \geq t_{br} ~:~ {\sf RCD}(p_s, p_j, t')= {\sf true} ~\text{with}~ t \geq t' + \zeta(e_{s, j},t') 
		~~~
		\text{(AK2)}\\ 
		1 & $ if $ \exists ~p_1,\dots,p_{k} ~:~ 
		~\forall i \in [1,k] , ~ \mathcal{A}_k (p_i,t_i) = 1 ~ \text{and} \\
		&	~~~~ \exists~ t'_i \geq t_{i} 
		~:~ {\sf RCD}(p_j, p_i, t'_i)= {\sf true} ~\text{with}~ 
		t \geq t'_{i}  + \zeta(e_{i, j},t'_i) ~~~\text{ (AK3)}\\
		0 & ~ otherwise
		\end{cases}$
	}
	
\end{definition}
\begin{definition}
Let $\mathcal{G}=( V,E,\rho,\zeta)$ be a TVG, and let $p_s$ be a node called \emph{source}. A \emph{temporal minimum $k$-level ordering} of $\mathcal{G}$ (TMKLO) from $p_s$ is a partition of the nodes in levels $L_{t_i}$ defined as follows:

$$p \in L_{t_i} ~ \emph{iff} ~ t_i={\sf min}~ t \in \mathbb{N} ~ \emph{such ~that} ~ \mathcal{A}_k(p, t_i)=1$$
\end{definition}
\noindent
Let us denote as $P_k$ the partition identifying the temporal minimum $k$-level ordering.
%
%\GF{in case, remember to stress about the bond between TMKLO and $t_{br}$, maybe with a subscript on TMKLO}

\begin{figure}[t]
	\centering
	\includegraphics[width=\linewidth]{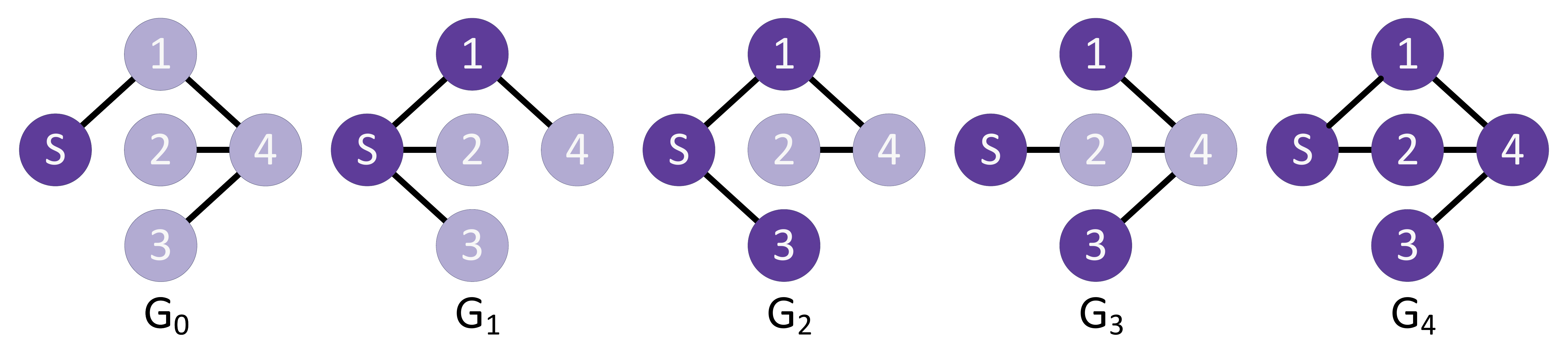}
	\caption{TVG example.}
	\label{fig:tvg5}
\end{figure}
%\begin{figure}[t]
%	\centering
%	\includegraphics[scale=0.5]{figures/Ak_TVG5}
%	\caption{$\mathcal{A}_{k=2}(p_s,t)$ of Figure \ref{fig:tvg5}}
%	\label{fig:ak_tvg5}
%\end{figure}

As an example, let us consider the TVG presented in Figure~\ref{fig:tvg5}: it evolves in five discrete time instants (\emph{i.e.}, $t_0, t_1, \dots, t_4$), its latency function $\zeta(e, t)$ is equal to $1$ for every edge $e$ at any time $t$.
%REMOVED
%and the upper bound on communication delays is $\delta\le1$.
Now, let us consider process $p_s$ as a source node that broadcasts $m$ at time $t_{br} = 0$, and let us assume that $k=2$. 
Such a TVG admits a temporal minimum $2$-level ordering $P_2=\{L_{t_0} = \{p_s\}, L_{t_1} = \{p_1\}, L_{t_2} = \{p_3\}, L_{t_4} = \{p_2,p_4\}\}$. Indeed:
\begin{itemize}
	\item The $2$-acceptance function $\mathcal{A}_2(p_s,t)$ is equal to $1$ for $t \geq t_{br} = t_0$ according to\emph{ AK1}. 
	\item The acceptance function evaluated on process $p_1$ is equal to $1$ for $t\geq 1$ according to {\emph{AK2}} (\emph{i.e.}, $t' = 0$ and $RCD(p_s,p_1,0)=true$ due to the presence function $\rho(<p_s, p_1>, \tau)=1, ~ \forall \tau \in [0, 1]$). 
	\item On processes $p_3$ and $p_2$, the acceptance function evaluates to $1$ respectively for $t \geq 2$ and for $t \geq 4$, for the same reasons as $p_1$. 
	\item The acceptance function on $p_4$ evaluates to $1$ for $t \geq 4$ according to {\emph{AK3}} (\emph{i.e.}, ${\sf RCD}(p_i, p_4, t'_i)={\sf true}$ for $p_i = p_1$, $t'_i = 1$, and for $p_i = p_3$, $t'_i = 3$). 
\end{itemize} 
%A graphical representation of the acceptance function for $p_s, p_1, p_2. p_3, p_4$ is presented in Figure~\ref{fig:ak_tvg5}.
%In the sequel, 
\vspace{10px}
%\NOTE{"I am wondering whether and how the snapshot times $t_0$, $t_1$, ... of the TVG and the times when the algorithm takes steps is related"}
%\NOTE{"I do not understand how the infinite retransmission of messages in DCPA helps if the sending of the messages is not synchronized with the TVGs."}[3cm]
We {now} present a sufficient condition (Theorem~\ref{th:liveness_sufficiency}) and a necessary condition (Theorem~\ref{th:liveness_necessity}) for the liveness of reliable broadcast based on the TMKLO.
%{\color{blue}
%	Due to the lack of space, we addr the proofs are reported in Appendix A.
%}

\begin{theorem}[DCPA liveness sufficient condition]\label{th:liveness_sufficiency}
Let $\mathcal{G}=( V,E,\rho,\zeta)$ be a TVG, let $p_s$ be the \emph{source} which broadcasts $m$ at time $t_{br}$, and let us assume $f$-locally bounded Byzantine failures. If there exists a partition $P_k=\{L_{t_{br}}, L_{t_1} \dots L_{t_x}\}$ of the nodes in $V$ representing a TMKLO of $\mathcal{G}$ associated to $m$ with $k>2f$, then the message $m$ spread using DCPA is eventually delivered by every correct process in $\mathcal{G}$.
\end{theorem}

\begin{proof}
We need to prove that if there exist a TMKLO with $k > 2f$ associated to message $m$, then any correct process eventually satisfies one of the CPA acceptance policies. A \emph{TMKLO} with $k > 2f$ implies that there exist a time $t$ such that the $2f+1$-acceptance function $\mathcal{A}_{k}(p,t)$ is equal to $1$ for every node of the network.

The process $p_s$ belongs to any TMKLO due to AK1: as the source of the broadcast, $p_s$ delivers the message according to AC1. Remind that the correct processes running DCPA spread the delivered messages over their neighborhood infinitely often. Then, the other nodes belong to the TMKLO due to the occurrence of AK2 or AK3.

If $AK2$ is satisfied by a node $p_j$ from time $t_j$, then $m$: \emph{(i)} can be delivered by the channel interconnecting $p_s$ with $p_j$ by definition of ${\sf RCD}()$, and \emph{(ii)} it is transmitted by $p_s$, because $t_j$ is greater than $t_{br}$. It follows that $p_j$ delivers $m$ according to $AC2$: indeed, $p_j$ has received $m$ directly from the source.

If $AK3$ is satisfied on a node $p_j$, it is possible to identify two scenarios: 
\begin{itemize}
\item {\bf Case 1}: ${\sf RCD}()$ is satisfied between $p_j$ and $2f+1$ nodes $p_i$ where $AK2$ is already satisfied.
We have shown that the processes satisfying $AK2$ accept $m$, and so they retransmit $m$. Assuming the $f$-locally bounded failure model, at most $f$ nodes among the neighbors of $p_i$ can be Byzantine and may not propagate $m$. Thus, $p_j$ receives at least $f+1$ copies of $m$ from distinct neighbors. According to $AC3$ of DCPA $p_j$ delivers $m$. 
\item {\bf Case 2}: ${\sf RCD}()$ is satisfied between $p_j$ and $2f+1$ nodes $p_i$ where $AK2$ or $AK3$ is already satisfied.
Inductively, as the nodes considered in {\bf Case 1} deliver $m$, it follows that the nodes $p_j$ satisfying $AK3$ due to at least $2f+1$ nodes $p_i$ where $AK2$ or $AK3$ already holds also deliver $m$.
\end{itemize}
\hfill$\Box$
\end{proof}
%\NOTE{"the necessary liveness condition for static networks is formulated positively (one needs an f+1-level ordering). The statement of Theorem \dots is negative."}[1cm]
\begin{theorem}[DCPA liveness necessary condition]\label{th:liveness_necessity}
	Let $\mathcal{G}=( V,E,\rho,\zeta)$ be a TVG, let $p_s$ be the \emph{source} that starts to broadcast $m$ at time $t_{br}$, and let us assume $f$-locally bounded Byzantine failures.
	The message $m$ can be delivered by every correct process in $\mathcal{G}$ only if a partition $P_k=\{L_{t_{br}}, L_{t_1} \dots L_{t_x}\}$ of nodes in $V$ representing a TMKLO of $\mathcal{G}$ associated to $m$ with $k>f$ exists.
%	If there does not exist any partition $P_k=\{L_{t_{br}}, L_{t_1} \dots L_{t_x}\}$ of nodes in $V$ representing a TMKLO of $\mathcal{G}$ associated to $m$ with $k>f$, then $m$ broadcast using DCPA \emph{cannot} be delivered by every correct process in $\mathcal{G}$.
\end{theorem}

\begin{proof}
	Let us assume for the purpose of contradiction that: \emph{(i)} every correct process in $\mathcal{G}$ delivers $m$, \emph{(ii)} the Byzantine failures are $f$-locally bounded, and \emph{(iii)} there does not exist a TMKLO associated to $m$ with $k > f$. The latter implies that the TMKLO with $k = f+1$ does not include all the nodes, \emph{i.e.} $\exists p \in \Pi ~|~ \forall t \in \mathbb{N}, \mathcal{A}_{f+1}(p,t) = 0$.
	
	The process $p_s$ is always included in a \emph{TMKLO} of any $k$. Thus, $p_s$ is included in $P_{f+1}$.
	The nodes that deliver $m$ according to $AC2$ have received $m$ from $p_s$. Thus, the ${\sf RCD}()$ predicate evaluated between $p_s$ and $p_i$ was true at least once after the delivery of $m$ by $p_s$. It follows that the condition defined in $AK2$ is eventually satisfied, and that those nodes are included in $P_{f+1}$.
	
	The remaining nodes that deliver according to $AC3$ have received the message from $f+1$ distinct neighbors.
	Let us initially assume that such neighbors have delivered the message by $AC2$. Again, the RCD predicate evaluated between the receiving node $p_j$ and the distinct $f+1$ neighbors $p_i$ has been true at least once after the respective deliveries of $m$. We already proved that such neighbors of $p_i$ are included in $P_{f+1}$, therefore the condition defined in $AK2$ is satisfied by those $p_j$ and they are included in $P_{f+1}$.
	
	It naturally follows that the remaining nodes (the ones that have received the message from neighbors satisfying $AC2$ or $AC3$) are included in $P_{f+1}$. This is in contradiction with the assumptions we made, because eventually every process satisfies one of the conditions AK1, AK2 or AK3, and the claim follows.	\hfill$\Box$
\end{proof}
%\NOTE{"the necessary and sufficient conditions given in Theorems 2 and 3 (formulated in terms of the existence of a TMKLO with certain properties) are very abstract and do not tell much about the properties of the underlying TVG."}
%{\color{blue}
%Due to the fact we employed a general network model as $TVG$, the conditions we defined can be verified describing the specific network as a $TVG$. As a matter of fact, any other liveness condition defined for CPA on dynamic networks relays (i) on the network topology and (i) on the time instants when the channels are available. 
%}

%
% !TEX root = main.tex
\section{On the Detection of DCPA Liveness}
\label{sec:detecting}

In Section \ref{sec:dcpa}, we proved that DCPA always ensure the reliable broadcast safety, and we provided the necessary and sufficient conditions about the dynamic network to enforce the reliable broadcast liveness. In this section, we are investigating the ability of individual processes to detect whether the reliable broadcast liveness is actually achieved in the current network. In more detail, we seek answers to the following questions:

\begin{itemize}
\item {\bf(Conscious Termination):} Given a message $m_s$ sent by a source $p_s$ on TVG $\mathcal{G}$, is $p_s$ able to detect if $m_s$ will eventually be delivered by every correct process? 
\item {\bf (Bounded Broadcast Latency):} Given a message $m_s$ sent by a source $p_s$ on TVG $\mathcal{G}$, is $p_s$ able to compute upper and lower bounds for reliable broadcast completion?
\end{itemize}
Obviously, if $p_s$ has no knowledge about $\mathcal{G}$, nothing about termination can be detected. As a consequence, some knowledge about $\mathcal{G}$ is required to enable Conscious Termination and Bounded Broadcast Latency.
We now formalize the notion of \emph{Broadcast Latency}, and introduce oracles that abstract the knowledge a process may have about $\mathcal{G}$.

\begin{definition} [Broadcast Latency (BL)]	
Let $\mathcal{G}=( V,E,\rho,\zeta)$ be a TVG and let $p_s$ be a node called \emph{source} that broadcasts a message $m$ at time $t_{br}$. We define as \emph{Broadcast Latency} $BL$ the period between $t_{br}$ and the time of the last delivery of $m$ by a correct process. 
\end{definition}
\noindent
We define the following knowledge oracles (from more powerful to least powerful):

\begin{itemize}
	\item \textbf{\emph{Full knowledge Oracle (FKO)}}: FKO provides full knowledge about the TVG, \emph{i.e.}, it provides $\mathcal{G}=( V,E,\rho,\zeta)$;
	
	\item \textbf{\emph{Partial knowledge Oracle (PKO)}}: given a TVG $\mathcal{G}=( V,E,\rho,\zeta)$, PKO provides the underlying static graph $G=(V, E)$ of $\mathcal{G}$;
	
	\item \textbf{\emph{Size knowledge Oracle (SKO)}}: given a TVG $\mathcal{G}=( V,E,\rho,\zeta)$, SKO provides the size of $\mathcal{G}$, that is $|V|$.
\end{itemize}
%REMOVED
%Due to the lack of space, the proofs for the Theorems that will be presented in this section are reported in Appendix A.

\subsection{Detecting DCPA Liveness on Generic TVGs}
%\NOTE{"Could one obtain similar results, when the source receives full knowledge about the graph infinitely often (or eventually)?"}
%{\color{blue}
In Section \ref{sec:dcpa} we showed that the conditions guaranteeing the liveness property of reliable broadcast are strictly bounded to the network evolution. It follows that the knowledge provided by an FKO, in particular about the network evolution starting from the broadcast time $t_{br}$, is necessary to argue on liveness, unless further assumptions are taken into account.
	In the following, we clarify how a process can employs an FKO to detect Conscious Termination and Bounded Broadcast Latency.
%}
%In the following we prove that assuming a generic TVG the FKO is  sufficient to detect Conscious Termination and Bounded Broadcast Latency while PKO is not.
%This is due to the fact that FKO enables a process to compute a TMKLO.
\begin{lemma}
	\label{lem:FKO-TMKLO}
%	\NOTE{"the construction in Lemma 1 analyses snapshots, so the FKO seems to give only the knowledge about the network configurations in past.\\
%		How large is T? Should it be delta * |V|?\\
%		How does one place the source in level $t_br$? How does the source find that all the snapshots have been analyzed?"}
%\GFN{I've changed nothing with respect those comments, do you think we should change anything? In particular, is it clear what $T$ is?}[6cm]
Let $\mathcal{G}=( V,E,\rho,\zeta)$ be a TVG, let $p_s$ be a node called \emph{source} that broadcasts a message $m$ at time $t_{br}$ and let us assume $f$-locally bounded Byzantine failures.
If $p_s$ has access to an FKO then it is able to verify if there exists a TMKLO for the current broadcast on $\mathcal{G}$.
\end{lemma}

\begin{proof}
%\renewcommand{\toto}{lem:FKO-TMKLO}
%\begin{color}{blue}
	In order to prove the claim it is enough to show an algorithm that verifies if a TMKLO exists, given the full knowledge of  the TVG provided by FKO.\\
	Such algorithm works as follow: initially, the source $p_s$ is placed in level $L_{t_{br}}$ of the TMKLO.
	Then, the snapshots characterizing the TVG have to be analyzed, starting from $G_{t_{br}}$ and following their order. In particular, for each snapshot $G_{t_{i}}$, $t_{i} \geq t_{br}$, we need to verify that: 
	\begin{enumerate}
	\item edges with only one endpoint already included in some level of the TMKLO are up enough to satisfy ${\sf RCD}()$ and 
	\item whenever ${\sf RCD}()$ is satisfied for a given edge $e_{i, j}$, we need to check if it allows $p_j$ to be part of the TMKLO as it satisfies one condition among AK2 and AK3.
	\end{enumerate}
	%Two associative arrays are populated, keeping track: (i) the duration of the edges (ii) and the TMKLO.
The algorithm ends when a TMKLO is found or when all the snapshots have been analyzed (and in the latter case we can infer that no TMKLO exists for the considered message on the given TVG). 
Assuming that $\mathcal{G}$ spans over $T$ time instants, the complexity of this algorithm is:
	$$O(|T||E) + O(|V|+|E|)  = O (|V| + |T||E|)$$	
\hfill$\Box$
%A more detailed description of the algorithm is delegated to the full version paper.
\end{proof}

\begin{theorem}
	\label{th:necFKO}
Let $\mathcal{G}=( V,E,\rho,\zeta)$ be a TVG, let $p_s$ be a node called \emph{source} that broadcasts a message $m$ at time $t_{br}$ and let us assume $f$-locally bounded Byzantine failures.
If $p_s$ has access to an FKO then it is able to detect if eventually every correct process will deliver $m$.
\end{theorem}

\begin{proof}
	The claim follows by considering that in order to assess the Conscious Termination of DCPA, the source process $p_s$ needs to  compute a TMKLO (i.e., it needs to check that eventually each correct process will be placed in a level) and due to Lemma \ref{lem:FKO-TMKLO} this can be done by accessing FKO. 
In particular, to detect Conscious Termination, a process $p_i$ can first verify if the necessary condition holds and this can be done by computing a TMKLO with $k \ge f+1$. If not, $p_i$ can simply infer that m will not be delivered by every correct process.
Contrarily, it can verify if the sufficient condition holds computing a TMKLO with $k \ge 2f+1$. If it exists, $p_i$ can infer that eventually every correct process will deliver the message otherwise, it needs to verify the necessary condition in every subgraph obtained by $\mathcal{G}$ removing all the possible disposition of Byzantine processes (remind that getting this answer corresponds to solve an NP-Complete problem even considering a static networks, 
thus the same intractability follows also on dynamic networks).
If the necessary condition is always satisfied, it can infer Conscious Termination otherwise not. \hfill$\Box$
\end{proof}

Let us note that if a process has the capability of computing the \emph{TMKLO} for a message $m$ sent at time $t_{br}$, then it can also establish a lower bound and an upper bound on the time needed by every correct process to deliver $m$ simply evaluating the maximum level of the TMKLO that satisfy respectively the necessary and the sufficient condition for DCPA.

\begin{theorem}
	\label{lem:time_bounds_low}
%	\GFN{I merged the two lemmas respectively for the lower and the upper bound in this single theorem}
	Let $\mathcal{G}=( V,E,\rho,\zeta)$ be a TVG and let $p_s$ be a node called \emph{source} that broadcasts a message $m$ at time $t_{br}$ and let us assume $f$-locally bounded Byzantine failures. Let $P_{f+1}=\{L_{t_0}, L_{t_1} \dots L_{t_x}\}$ be the TMKLO with $k= f+1$ associated to $m$ and let $t_{max}^{f+1}$ be the time associated to the last level of $P_{f+1}$. 
%	The computed TMKLO provides a lower bound for BL such that:
%	
%	$$t_{max}^{f+1} -t_{br} \le BL$$
	Let	assume the existence of the TMKLO with $k= 2f+1$ associated to $m$, 		
	$P_{2f+1}=\{L_{t_0}, L_{t_1} \dots L_{t_x}\}$,   and let $t_{max}^{2f+1}$ be the time associated to the last level of $P_{2f+1}$. 
	The computed TMKLOs provide respectively a lower bound and an upper bound for BL such that:
	
%	$$ BL \le t_{max}^{2f+1} -t_{br}$$	
$$t_{max}^{f+1} -t_{br} \le BL \le t_{max}^{2f+1}-t_{br}$$
\end{theorem}

\begin{proof}
\emph{Lower Bound:}
Let us assume for the purpose of contradiction that BL can be lower than $t_{max}^{f+1} -t_{br}$. It follows that the last process $p_i$ delivering $m$ does it at a time $t_{i} < t_{max}^{f+1}$.
Given the definition of TMKLO with $k=f+1$, a level $L_x$ is created each time that a process not yet inserted in the TMKLO delivers a message (due to AK2 or AK3). As a consequence, the last level of the TMKLO is created when the last process delivers the message. Thus, considering that $p_i$ is the last process delivering the message, it follows that $t_i$ is the time associated to the last level.
Given $P_{f+1}$, it follows that $t_i = t_{max}^{f+1}$  and we have a contradiction.

\emph{Upper Bound:}
Let us assume for the purpose of contradiction that BL can be greater than $t_{max}^{2f+1} -t_{br}$. It follows that the last process $p_i$ delivering $m$ does it at a time $t_{i} > t_{max}^{2f+1}$.
Given the definition of TMKLO with $k=2f+1$, a level $L_x$ is created each time that a process not yet inserted in the TMKLO delivers a message (due to AK2 or AK3). As a consequence, the last level of the TMKLO is created when the last process delivers the message.
Thus, considering that $p_i$ is the last process delivering the message, it follows that $t_i$ is the time associated to the last level.
Given $P_{2f+1}$, it follows that $t_i = t_{max}^{2f+1}$  and we have a contradiction.\hfill$\Box$
\end{proof}

Remind that, as the sufficient condition we provided is not strict, a TMKLO with $k= 2f+1$ could not exist even if the reliable broadcast is achievable. It is also possible to provide a stricter upper bound for BL as we explained inside the proof of Theorem \ref{th:necFKO}, but is not practical to compute.	
Finally, let us remark that the knowledge on the underlying topology is not enough on dynamic networks to argue on liveness.

\begin{remark}
	\label{th:weaker_oracle}
	Let $\mathcal{G}=( V,E,\rho,\zeta)$ be a TVG and let $p_s$ be a node called \emph{source} that broadcasts a message $m$ at time $t_{br}$ and let us assume $f$-locally bounded Byzantine failures. 
If a process $p_s$ has access only to a PKO (and not to an FKO) then it is not able to detect either Conscious Termination and Bounded Broadcast Latency.
%	Given a dynamic network, let $p$ be a process there included. It is not possible for a process $p$ which endows \emph{PKO} to argue on the liveness property of DCPA.
Indeed, as we highlighted in section \ref{subsec:cpalive}, moving on dynamic network the knowledge on the underlying graph is not enough, because specific sequences of edge appearances are required in order to guarantee the message propagation (let us take again Figure \ref{fig:TVG} as clarifying example). Thus, a PKO is not enough in arguing on liveness.
%	\begin{color}{blue}
The same can be said about Bounded Broadcast Latency as PKO provides no information about the time instants when the edges will appear.
\end{remark}

%\begin{proof}	
%	It has been shown in \cite{litsas2013graph,pagourtzis2017reliable} the necessary and sufficient conditions to argue on the liveness property of CPA.
%	%on static networks where $f$-locally bounded failures are present. 
%	Such conditions are topological. As it is not possible to argue on liveness of CPA on static networks without the knowledge of the topology, it follows the same also on dynamic networks.
%However, moving on dynamic network, the knowledge on the underlying graph is not enough, because specific sequences of edge appearances are required in order to guarantee the message propagation (let us take again Figure \ref{fig:TVG} as clarifying example). Thus, PKO is not enough arguing on liveness.
%%	\begin{color}{blue}
%	The same can be said about Bounded Broadcast Latency as PKO provides no information about the time instants where the edge will appear.
%%	\end{color}
%%	\renewcommand{\toto}{th:weaker_oracle}
%\end{proof}
%
%In this section, we shown that \emph{FKO} is sufficient to detect liveness while PKO is not. 
%Given that the complete knowledge on network topology is required by CPA to argue on liveness on static networks, we conjecture that it is not possible to detect liveness in generic TVG with an oracle weaker than \emph{FKO}. 
\subsection{Detecting DCPA Liveness on Restricted TVGs}
\label{subsec:livespecialtvg}
\noindent
Casteigts \emph{et al.}~\cite{casteigts2012time} defined a hierarchy of TVG classes based on the strength of the assumptions made about appearance of edges. So far, we considered the most general TVG\footnote{Class $1$ TVG according to Casteigts \emph{et al.}~\cite{casteigts2012time}}. In the following, we consider two more specific classes of the hierarchy where we show that liveness can be detected using oracles weaker than FKO. 
In particular, we consider the following classes that are suited to model recurring networks:
\begin{itemize}
	\item \textbf{\emph{Class recurrence of edges, ER }}: if an edge $e$ appears once, it appears infinitively often\footnote{Class 6 TVG in Casteigts \emph{et al.}~\cite{casteigts2012time}.}.	
	\item \noindent\textbf{\emph{Class time bounded recurrences, TBER}}: if an edge $e$ appears once, it appears infinitively often and there exist an upper bound $\Delta$ between two consecutive appearances of $e$\footnote{Class 7 TVG in Casteigts \emph{et al.}~\cite{casteigts2012time}.}.	
\end{itemize}
Let us recall that assuming predicate ${\sf RCD}(e_{i, j}, t)={\sf true}$ for every edge $e_{i, j}$ at some time $t$ is necessary to guarantee liveness. While considering classes ER and TBER, such condition must be satisfied infinitely often, otherwise it is easy to show that the results presented in the previous section still apply.
Let us also note that the conditions we defined in Section \ref{subsec:cpalive} are related to a single broadcast generated by a specific source $p_s$ i.e., for a source $p_s$ broadcasting a message at time $t_{br}$ the conditions must hold from $t_{br}$ on. 
Contrarily, exploiting the recurrence of edges it is possible to define different conditions that are valid for every broadcast from the same source  $p_s$, independently from when it starts.
\subsubsection{Detecting DCPA Liveness in ER TVG}

In this section, we prove that considering TVG of class ER, we can get the following results: \emph{(i)} PKO (an oracle weaker than FKO) is enough to enable Conscious Termination, \emph{(ii)} despite the more specific TVG considered, FKO is still required to establish upper bounds for \emph{BL}. Intuitively, this results follows from the fact that PKO allows to determine whether a MKLO exists on the static underlying graph, and this is enough to detect if eventually every correct process will be able to deliver the message. However, given the absence of information on when each edge is going to appear, it is impossible to compute an upper bound on the time required to accomplish the broadcast.

\begin{lemma}\label{lem:PKO-MKLO}
	%	\GFN{Consider to remove this lemma and to substitute with a sentece}
	Let $\mathcal{G}=( V,E,\rho,\zeta)$ be a TVG and 
	%of class ER that ensures ${\sf RCD}()$ infinitively often, let $p_s$ be a node called \emph{source} that broadcasts a message $m$ at time $t_{br}$ and let us assume $f$-locally bounded Byzantine failures. 
	let $G=(V, E)$ be the associated underlying graph.
	If $p_s$ has access to a PKO then it can compute a MKLO on $G$.
\end{lemma}

\begin{proof}
	The PKO provides knowledge on the topology of $G$.
	We reminded in Definition 1 that the MKLO is a partition of the nodes on the base of a topological conditions. 
	It follows that it is possible to verify the MKLO on $G$ with PKO through a modified breath-first search \cite{litsas2013graph}.
	\hfill$\Box$
\end{proof}

\begin{lemma}
	\label{lem:particular_cases_ER}	
	Let $\mathcal{G}=( V,E,\rho,\zeta)$ be a TVG of class \emph{ER} that ensures ${\sf RCD}()$ infinitively often, let $G=(V, E)$ be the static underlying graph of $\mathcal{G}$, let $p_s$ be a node called \emph{source} and let us assume $f$-locally bounded Byzantine failures. 
	If there exists the \emph{MKLO} of $G=(V, E)$ associated to $p_s$ then there always exists the \emph{TMKLO} of $\mathcal{G}$ associated to a message $m$ sent by $p_s$ with the same $k$.
\end{lemma}

\begin{proof}
	We prove the claim showing a mapping from MKLO to TMKLO.
	The source is placed inside the \emph{TMKLO} at level $t_{br}$. Then, given the assumption on the channels and that every node in the MKLO has either (i) an edge connecting it with the source (ii) and/or $k$ neighbors already included in MKLO, it follows that every node eventually satisfies at least one between AK2 and AK3.
	\hfill$\Box$
\end{proof}

\begin{theorem}
	\label{th:necPKO}
Let $\mathcal{G}=( V,E,\rho,\zeta)$ be a TVG of class ER that ensures ${\sf RCD}()$ infinitively often, and let $p_s$ be a node called \emph{source} that broadcasts $m$ at time $t_{br}$, and let us assume $f$-locally bounded Byzantine failures.
If $p_s$ has access to a PKO, then it is able to detect if eventually every correct process delivers $m$.
\end{theorem}

\begin{proof}
It follows from Lemma \ref{lem:PKO-MKLO} and Lemma \ref{lem:particular_cases_ER}\hfill$\Box$
\end{proof}
\subsubsection{Detecting DCPA Liveness in TBER TVG}
The liveness condition enabling CPA to enforce reliable broadcast relays on the network topology, therefore an oracle weaker that FKO cannot enable Conscious Termination unless further assumptions are made. On the other hand, the weaker oracle SKO allows a process to compute Bounded Broadcast Latency.

\begin{lemma}\label{lem:c7_1}
	Let $\mathcal{G}=( V,E,\rho,\zeta)$ be a TVG of class \emph{TBER} where each edge $e_{i, j}$ reappears in at most $\Delta$ time instants satisfying ${\sf RCD}(e_{i, j}, t)$. Let $\delta_{max} = \text{max} (\zeta(e,t))$. Let $p_s$ be the source and let us assume $f$-locally bounded Byzantine failures. The Broadcast Latency BL is upper bounded by	
	$$ BL \le |V|(\delta_{max} + \Delta)$$
\end{lemma}

\begin{proof}
	%\GF{Check}
	%
	Given the assumptions on the TVG, we know that every edge reappears in $\Delta$ and satisfies ${\sf RCD}()$.
	The worst case scenario, with respect the message propagation, is the one in which every node has to wait $\Delta$ to forward a message.
	The worst case scenario, with respect the network topology, is the one where every process has to wait the last one which has delivered to deliver (in other words, the partitions of the MKLO evaluated over the underlying graph $G(V,E)$, with the exception of the second level, have size equals to 1).
	%	We depicted such a worst case scenario in Figure \ref{fig:th_c7}.
	%
	%	\renewcommand{\toto}{lem:c7_1}
	\hfill$\Box$
\end{proof}

%\begin{figure}
%	\centering
%	%\begin{subfigure}{.5\textwidth}
%	%		\centering
%	\includegraphics[scale=0.5]{figures/th_c7}
%	\caption{Worst case topology with respect BL.}
%	\label{fig:th_c7}
%\end{figure}

\begin{lemma}\label{lem:c7_2}
	Let $\mathcal{G}=( V,E,\rho,\zeta)$ be a TVG of class \emph{TBER} where each edge $e_{i, j}$ reappears in at most $\Delta$ time instants
	{
		satisfying ${\sf RCD}(e_{i, j}, t)$. 
		Let $\delta_{max} = \text{max} (\zeta(e,t))$.
	}	
	%	 time instants, its latency function $\zeta$ is bounded by $\delta$, and it satisfies the ${\sf RCD}(e_{i, j}, t)$ periodically. 
	%for every edge ($ \forall e_{i, j} \in E, ~ \forall t \in \mathbb{N}, ~ \zeta(e_{i, j},t) \le \delta$) and 
	%	and they always enforce RCD. 
	Let $p_s$ be the source and let us assume $f$-locally bounded Byzantine failures. 
	Let $P_{2f+1}=\{L_{t_0}, L_{t_1} \dots L_{t_x}\}$ be the MKLO with $k= 2f+1$ computed on the underlying graph $G=(V, E)$ 
	(if exists)	and let $S_{2f+1}$ be size of $P_{2f+1}$. 
	%	The computed TMKLO provides an upper bound for BL such that:
	
	An upper bound for BL can be computed from the MKLO with $k=2f+1$. In particular:
	
	$$BL \le S_{2f+1}(\delta_{max} + \Delta)$$	
\end{lemma}

\begin{proof}
	Given the assumptions on the TVG $\mathcal{G}$ we know that every edge reappears in $\Delta$ and guarantees ${\sf RCD}()$ 
	%in $\delta$ time instants
	.
	The worst case scenario with respect the message propagation is the one where every node as to wait $\Delta$ to forward a message.

	The bound follows by Theorem \ref{lem:time_bounds_low} and Lemma \ref{lem:particular_cases_ER}, noting that every node in level $L_i$ delivers in $(\delta_{max} + \Delta)i$ time instants.
	\hfill$\Box$
\end{proof}

\begin{theorem}\label{th:necFKOlat_TBER}
Let $\mathcal{G}=( V,E,\rho,\zeta)$ be a TVG of class \emph{TBER}  where each edge $e_{i, j}$ reappears in at most $\Delta$ time instants
satisfying ${\sf RCD}(e_{i, j}, t)$. 
Let $\delta_{max} = \text{max} (\zeta(e,t))$. Let $p_s$ be a node called \emph{source} that broadcasts $m$ at time $t_{br}$, and let us assume $f$-locally bounded Byzantine failures. Let $P_{2f+1}=\{L_{t_0}, L_{t_1} \dots L_{t_x}\}$ be the MKLO with $k= 2f+1$ associated to $m$ and computed on the underlying graph $G=(V, E)$ (if exists) and let $S_{2f+1}$ be size of $P_{2f+1}$. 
If $p_s$ uses SKO or PKO, then $p_s$ is able to compute an upper bound for BL.
%If $p_s$ has access to a PKO then it is able to compute an upper bound and a lower bound for BL.
Specifically:
$$ BL \le |V|(\delta_{max} + \Delta)) \text{using SKO}$$
$$ BL \le S_{2f+1} (\delta_{max} + \Delta) \text{using PKO}$$
\end{theorem}
\begin{proof}
The claim follows from Lemmas \ref{lem:c7_2} and \ref{lem:c7_1}.\hfill$\Box$
\end{proof}

% !TEX root = main.tex
\section{Moving to an Asynchronous System}
\label{sec:asy}
In this work we assumed a synchronous distributed systems.
In this section, we briefly discuss consequences of asynchrony on the safety and liveness of DCPA.

In Section \ref{subsec:cpasafe}, we showed that a reliable and authenticated channel is necessary and sufficient to enforce safety 	through CPA in an \emph{f-locally bounded} failure model. 
Such channel properties are independent of the latency function. Indeed, they require that if a message $m$ sent by a correct process is eventually received at its destination, it has not been compromised by the channel.
As a consequence CPA (and DCPA as well) continues to enforce safety also on asynchronous dynamic networks.

In Section \ref{subsec:cpalive}, we pointed out the need of having channels up long enough to allow the delivery of messages. This imposes constraints on the presence function due to the latency function.
The asynchrony affects the latency function $\zeta(e, t)$ that basically is no more bounded.
%Concerning liveness, we underlined in subsection \ref{subsec:cpalive} the necessity of identify time instants where the channels enforce a reliable behavior. 
This makes impossible (in asynchronous system) to establish constraints for the liveness due to the fact it is no longer guaranteed the propagation of messages. 
It follows that we cannot argue on liveness of reliable broadcast on general TVG without making further assumptions.

In Section \ref{subsec:livespecialtvg} we investigated about liveness in specialised classes of TVG. In particular, we showed in Theorem \ref{th:necPKO} that assuming recurrent \emph{RCD} and having the knowledge on the underlying static graph it is possible to investigate about. 
It follows that, although RCDs are not identifiable over the time, if they are satisfied infinitively often, they enable the verification of liveness also in asynchronous systems.
\section{Conclusion}
% !TEX root = main.tex
\noindent
We considered the reliable broadcast problem in dynamic networks represented by TVG. 
We analyzed the porting conditions enabling CPA to be correctly employed on dynamic networks. The analysis of this simple algorithm is important as it works exploiting only local knowledge.
This contrasts to the best result so far in the same setting~\cite{maurer2015communicating}, that demands an exponential costs to check when a message can be delivered. 
Moreover, we presented necessary and sufficient conditions to ensure safety and liveness DCPA. We analyzed how much knowledge of the TVG is needed to detect whether the liveness condition is satisfied, and its cost. 
Our work is a starting point to identify more general parameters of dynamic networks that guarantees the fulfillment of the conditions we provided, both in a deterministic and probabilistic way.
Other interesting points to address in future works are: i) the definition of a more realistic locally bounded failure model that takes also the time dimension into account, ii) the research of conditions on the dynamic network enabling nodes to conscious termination with just local information. 

\bibliographystyle{splncs04}
\bibliography{references}

%\begin{subappendices}
%	\renewcommand{\thesection}{\arabic{section}}%
%%	\input{appendixes}
%	\input{asynchronous}
%%	\input{retransmissions}
%%	\input{nonuniform}
%\end{subappendices}

\end{document}